\documentclass[letterpaper, 10 pt, conference]{ieeeconf}  
\usepackage{graphicx}
\usepackage{graphics} 
\usepackage{epsfig} 
\usepackage[abs]{overpic}
\usepackage{algorithmic}
\usepackage{amsmath}
\usepackage{amsfonts}

\usepackage{tikz}
\usepackage{theorem}
\newtheorem{theorem}{\bf{Theorem}}[section]

\newtheorem{cor}[theorem]{Corollary}

\newtheorem{prop}[theorem]{Proposition}
\theoremstyle{plain}

\newenvironment{definition}[1][Definition]{\begin{trivlist}
\item[\hskip \labelsep {\bfseries #1}]}{\end{trivlist}}

\title{\LARGE \bf
A Tight Lower Bound on the Controllability of \\
Networks with Multiple Leaders
}

\author{A. Yasin Yazicioglu, Waseem Abbas, and Magnus Egerstedt\\
Electrical and Computer Engineering\\
Georgia Institute of Technology, Atlanta, GA 30332, USA\\
{\tt\small yasin@ece.gatech.edu, wabbas@gatech.edu, magnus@gatech.edu}
}

\begin{document}

\maketitle
\begin{abstract}
In this paper we study the controllability of networked systems with static network topologies using tools from algebraic graph theory. Each agent in the network acts in a decentralized fashion by updating its state in accordance with a nearest-neighbor averaging rule, known as the consensus dynamics. In order to control the system, external control inputs are injected into the so called leader nodes, and the influence is propagated throughout the network. Our main result is a tight topological lower bound on the rank of the controllability matrix for such systems with arbitrary network topologies and possibly multiple leaders.
\end{abstract}


\section{Introduction}
Decentralized control of networked multi-agent systems has received a considerable amount of attention during the last decade. Numerous applications of decentralized control laws have been studied including flocking (e.g., \cite{Jadbabaie07}), alignment and formation control (e.g., \cite{Fax02}-\cite{Lin04}), distributed estimation (e.g., \cite{Speranzon06}), sensor coverage (e.g., \cite{Cortes04}) and distributed control of robotic networks (e.g., \cite{Bullo09}), to name a few. In a distributed framework, a global task is achieved by the local interactions of agents among each other without a centralized control. Some tasks only require a common agreement between the agents, whereas others may ask for agents to achieve some definite configuration in terms of their defined states. For such tasks, a fundamental question is whether such a decentralized system can be controlled by directly manipulating only some of the agents. This question motivates our analysis of the controllability of networked systems. 

Controllability of networked systems was initially addressed in \cite{Tanner04}, where a connection between the spectral properties of the underlying graph modelling a network, and the controllability of the system was analyzed. A more topological analysis of the problem was later presented in \cite{Rahmani06} with an emphasis on how the symmetry with respect to the leader node affects the controllability of the system. More general conditions were presented in \cite{Ji07,Martini08} by introducing equitable partitions in the analysis. These concepts were extended along with additional results in \cite{Rahmani09}. In \cite{Egerstedt10}, these equitable partitions were used to obtain an upper bound on the rank of the controllability matrix. Recently, distance partitions are used in \cite{Zhang11} to obtain a lower bound on the rank of the controllability matrix for single-leader networks. 
%

In this paper we analyse leader-follower networks in which the agents utilize a nearest-neighbor averaging rule. Some agents, called the \textit{leaders,} support external control inputs that ultimately influence the dynamics of all other agents namely \textit{followers} by spreading throughout the network. We explore the controllability of the overall  system under this setting. Our main result is a topological lower bound on the rank of the controllability matrix for any graph structure with multiple leaders. This lower bound is based on distances of nodes to the control nodes, and it can easily be computed directly from the network topology, without having to rely on any rank test or spectral analysis of the graph. This problem was studied for single-leader networks in \cite{Zhang11}, and a lower bound was obtained using the distance partition with respect to the leader. In this work, we tackle the general problem with possibly multiple leaders by extending the use of distance based relationships to such cases.

The organization of this paper is as follows: Section \ref{prelim} presents some preliminaries related to the system dynamics and algebraic graph theory. In Section \ref{controllability}, we present our controllability analysis. Section \ref{compute} provides an algorithm to compute the proposed lower bound on the rank of the controllability matrix for arbitrary networks. Finally, Section \ref{conclusion} provides the concluding remarks.

\section{ Preliminaries}
\label{prelim}
Consider a networked system of $n$ agents that utilize the same nearest neighbor averaging rule, known as the consensus equation, to govern their dynamics. For each particular agent $i$, the consensus equation is given as 
\begin{equation}
\label{consensus}
\dot{x}_i=\sum_{j\in \mathcal{N}_i}(x_j-x_i),
\end{equation}
where $x_i$ is the state of agent $i$, and  $\mathcal{N}_i$ is the set of agents neighboring agent $i$. Without loss of generality, let us assume that $x_i \in \mathbb{R}$, and the interactions among the agents are encoded via a static undirected graph $\mathcal{G}=(V,E)$. In this graph, each node in the node set, $V=\{1,2,\hdots,n\}$, corresponds to a particular agent, and the edge set, $E \subseteq V \times V$, is the set of unordered pairs $(i,j)$ depicting that the nodes $i$ and $j$ are neighbors. In this context, neighbor nodes are the ones that have the measurements of each other's states. 

The consensus equation provides a simple, yet powerful foundation for decentralized control strategies that can be utilized in various tasks, including coverage control, containment control, distributed filtering, flocking and formation control. With all agents utilizing the consensus equation, their states asymptotically converge to the stationary mean,
if and only if the underlying graph is connected \cite{Jadbabaie03}. 

Assume that we would like to control this network simply by applying external control signals to some of the nodes. Without loss of generality, let the first $m$ nodes be the  leaders taking the external control inputs, and let the remaining $(n-m)$ nodes be the followers whose dynamics are governed by (\ref{consensus}). Let the $m$ dimensional control input be represented by vector $u$. Then, the dynamics of the leader nodes satisfy
\begin{equation}
\label{control}
\dot{x}_{i}=\sum_{j\in \mathcal{N}_{i}}(x_j-x_{i})+[u]_i, \mbox{ for $i=1,2\hdots,m$.} 
\end{equation}
where, $[u]_i$ denotes the $i^{th}$ entry of the control vector $u$. When the external control signals are applied to the leader nodes, their effect on the dynamics propagates to the rest of the nodes through the underlying network. 

Our main goal here is to characterize the controllability of the overall system under this setting. In particular, we are interested in the dimension of the controllable subspace, and aim to relate it to the topology of the underlying network from a purely graph theoretic perspective. To this end, we use some basic tools from algebraic graph theory, in particular the degree matrix, the adjacency matrix, and the graph Laplacian.

Let $\Delta$ be the $n \times n$ degree matrix associated with the graph. The entries of $\Delta$ are given as
\begin{equation}
\label{Deg}
[\Delta]_{ij}=\left\{\begin{array}{ll}|\mathcal{N}_i|&\mbox{ if } 
i=j\\0&\mbox{ otherwise, }\end{array}\right.
\end{equation}
where $|\mathcal{N}_i|$ denotes the cardinality of  $\mathcal{N}_i$, and it is equal to the number of neighbors of node $i$.

The adjacency matrix, $\mathcal{A}$, is an $n \times n$ symmetric matrix with its entries given as

\begin{equation}
\label{Adj}
[\mathcal{A}]_{ij}=\left\{\begin{array}{ll}1&\mbox{ if } 
(i,j)\in E \\ 0&\mbox{ otherwise. }\end{array}\right.
\end{equation}

The graph Laplacian, $L$, is simply given as the difference of the degree and the adjacency matrices,
\begin{equation}
\label{Laplacian}
 L = \Delta-\mathcal{A}. 
\end{equation}

In light of (\ref{consensus}) and (\ref{control}), the dynamics of the leader-follower network with $m$ leaders can be given as 
\begin{equation}
\label{sys}
\dot{x}=-Lx+Bu,
\end{equation}
where $x=[x_1, x_2, \hdots, x_n]^T$ is the state vector obtained by stacking the states of each individual node, and $B$ is an $n \times m$ matrix  with the following entries
\begin{equation}
\label{Beq}
[B]_{ij}=\left\{\begin{array}{ll}1&\mbox{ if $i=j $} 
 \\ 0&\mbox{ otherwise. }\end{array}\right.
\end{equation}
Note that (\ref{sys}) represents a standard linear time-invariant system and it relates the system dynamics to the graph topology through the graph Laplacian.  

\section{Controllability of Leader-Follower Networks}
\label{controllability}
In this section we will analyse the controllability of the system given in (\ref{sys}). 
In particular, we present relationships between the network topology and the rank of the controllability matrix for such systems. We start this section by referring to the results based on the equitable partitions presented in \cite{Martini08,Egerstedt10}. 

A \emph{partition} of a graph $\mathcal{G}=(V,E)$ is given by a mapping $\pi : V \rightarrow \{C_1,C_2,\hdots,C_r\}$, where $\pi(i)$ denotes the cell that node i gets mapped to, and we use $dom(\pi)$ to denote the domain to which $\pi$ maps, i.e., $dom(\pi) = \{C_1,C_2,\hdots,C_r\}$.

\begin{definition}
(External Equitable Partition): A partition $\pi$ of a graph $\mathcal{G}$ with cells $C_1,C_2, \hdots, C_r$ is said to be an \emph{external equitable partition} (EEP) if each node in cell $C_i$ has the same number of neighbors in cell $C_j$ for every $i\neq j$.
\end{definition}

In the controllability analysis, we are particularly interested in the \emph{maximal leader-invariant} EEP of a graph. An EEP is said to be \emph{leader-invariant} if the leader nodes are mapped to singleton cells, and such a mapping is said to be \emph{maximal} if no other leader-invariant EEP with fewer number of cells exists. Note that for any graph $\mathcal{G}$ there is a unique maximal leader-invariant EEP, $\pi^*$. Examples of maximal leader-invariant EEPs are depicted in Fig. \ref{EEPfig}.

\begin{figure}[htb]
\centerline{\epsfig{figure=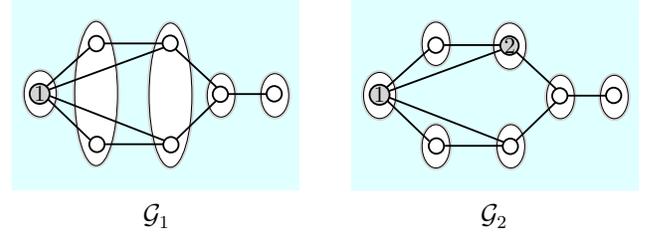,height=1.2in,width=3.3in}}
\caption{\label{fig}Maximal leader-invariant EEP's for two networks, $\mathcal{G}_1$ and $\mathcal{G}_2$. $\mathcal{G}_1$ has a single leader namely $1$, whereas $\mathcal{G}_2$ has two leaders namely $1$ and $2$.}
\label{EEPfig}
\end{figure}

Maximal leader-invariant EEPs are useful structures in the controllability analysis since the states of the nodes that appear in the same cell of the maximal leader-invariant EEP asymptotically converge to the same value \cite{Martini08}.
\begin{theorem}
\label{collapse} \cite{Martini08}
 If $\mathcal{G}$ is a connected graph with $\pi^*$ being its maximal leader-invariant EEP, then for all $C_i \in dom(\pi^*)$
\begin{equation}
\label{LEPlimit}
 \lim_{t \to \infty}(x_k(t)-x_l(t))=0, \forall k,l \in \pi^{*-1}(C_i).
 \end{equation}
\end{theorem}

In light of Theorem \ref{collapse}, 
one can at most be able to control all of the average state values within each cell of the $\pi^*$. Hence, the cardinality of $dom(\pi^*)$ provides an \emph{upper bound} on the rank of the controllability matrix as given in \cite{Egerstedt10}.

\begin{theorem}
\label{upperbound} 
 \cite{Egerstedt10}
Let $\mathcal{G}$ be a connected network, and $\pi^*$ denote its maximal leader-invariant EEP. Given the dynamics in (\ref{sys}), the rank of the controllability matrix, $\Gamma$, satisfies
\begin{equation}
\label{ubound}
rank(\Gamma)\leq |\pi^*|,
\end{equation}
where $|\pi^*|$ is the cardinality of $dom(\pi^*)$.
\end{theorem}

The upper bound given in Theorem \ref{upperbound} is quite useful in analyzing the controllability of a leader-follower network. For instance, one can conclude that a system is not completely controllable if there exists non-singleton cells in its maximal leader-invariant EEP. However, all the cells being singletons does not necessarily imply that the network is completely controllable.  
 
Next, we present our main result, a \emph{lower bound} on the rank of the controllability matrix when multiple leaders are present. In \cite{Zhang11}, the authors present a lower bound for single-leader networks. To this end, they utilize the distance partition of an underlying graph with respect to its leader. In this partition all the nodes that are at the same distance from the leader are mapped into a single cell. It is shown there that the rank of the controllability matrix is greater than or equal to the number of cells in this partition. 
\begin{theorem}
\label{singlelower} 
 \cite{Zhang11}
Let $\mathcal{G}$ be a connected single-leader network, and $\pi_D$ denotes its distance partition with respect to the leader. Given the dynamics in (\ref{sys}), the rank of the controllability matrix, $\Gamma$, satisfies
\begin{equation}
\label{ubound}
|\pi_D| \leq rank(\Gamma),
\end{equation}
where $|\pi_D|$ is the number of cells in the distance partition.
\end{theorem}

Similar to the single-leader case, the distances of nodes from the leaders appear as the fundamental property in our analysis. We start our analysis with the following proposition. 

\begin{prop}
\label{prop1} Let $\mathcal{G}=(V,E)$ be a connected network with the dynamics in (\ref{sys}), and let $b_k$ be the $k^{th}$ column of the input matrix $B$. Then, for any node $i$ and leader $k$,
\begin{equation}
\label{propeq}
[(-L)^r b_k]_{i}=\left\{\begin{array}{ll}0&\mbox{ if $0\leq r < d_{ik} $} 
 \\ \mbox{$[\mathcal{A}^{r}]_{ik}$}&\mbox{ if $r = d_{ik} $ }\end{array}\right.
\end{equation}
where $L$ is the graph Laplacian, $\mathcal{A}$ is the adjacency matrix of the graph, and $d_{ik}$ is the distance of node $i$ to the leader node $k$. \end{prop}
\begin{proof}
\newline 
Using the equality in (\ref{Laplacian}), $(-L)^r$ can be expanded as
 \begin{equation}
\label{A_power}
(-L)^r=(\mathcal{A}-\Delta)^r=\mathcal{A}^r+\sum_{m=0}^{r-1}(-1)^{r-m}\mathcal{S}_{m},
\end{equation}
where $\mathcal{S}_{m}$ is the sum of all matrices that can be represented as a multiplication in which $\mathcal{A}$ appears $m$ times and $\Delta$ appears $r-m$ times. 
Note that since $\Delta$ and $\mathcal{A}$ have only non-negative entries, any matrix that can be represented this way has only non-negative entries. Moreover, since $\Delta$ is a diagonal matrix with positive entries on the main diagonal, it doesn't add or remove zeros when multiplied by a matrix. Hence, $\mathcal{S}_{m}$ has  zeros only at the same locations as $\mathcal{A}^{m}$, and the following condition is satisfied:
 \begin{equation}
\label{SA}
[\mathcal{S}_{m}]_{ik}=0 \iff [\mathcal{A}^{m}]_{ik}=0.
\end{equation}
Using (\ref{Beq}) and (\ref{A_power}), the $i^{th}$ entry of the vector $(\mathcal{A}-\Delta)^rb_k$ can be expressed as follows:
 \begin{eqnarray}
\label{A_kb_i}
[(\mathcal{A}-\Delta)^rb_k]_i&=&[(\mathcal{A}-\Delta)^{r}]_{ik} \nonumber \\
&=&[\mathcal{A}^{r}]_{ik} +\sum_{m=0}^{r-1}(-1)^{r-m} [S_{m}]_{ik}.
\end{eqnarray}
As $\mathcal{A}$ is the adjacency matrix of the graph, $[\mathcal{A}^{r}]_{ik}$ is equal to the number of paths of length $r$ from node $i$ to node $k$. Since the distance of node $i$ to the leader node $k$ is $d_{ik}$, $[\mathcal{A}^{r}]_{ik}=0$ for all $0\leq r<d_{ik}$. Hence, (\ref{SA}) and (\ref{A_kb_i}) together imply that $[(\mathcal{A}-\Delta)^rb_k]_i=0$ for all $0\leq r<d_{ik}$. Furthermore, plugging $r=d_{ik}$ into (\ref{A_kb_i}), we get
 \begin{eqnarray}
\label{A_lb_i}
[(\mathcal{A}-\Delta)^{d_{ik}}b_k]_i=[\mathcal{A}^{d_{ik}}]_{ik},
\end{eqnarray}
where $[\mathcal{A}^{d_{ik}}]_{ik}$ is equal to the number of paths with the shortest length, $d_{ik}$, from node $i$ to the leader node $k$, and for a connected graph it is non-zero. 
 
 \end{proof}
   
In a network with $m$ leaders, for each node $i$ we can define an $m$ dimensional distance vector, $d_i$, that contains the distance of node $i$ to each of the leaders as
 \begin{equation}
\label{d_vec}
d_i=\left[ \begin{array}{ccccc}
 d_{i1}  & d_{i2}  &  \hdots & d_{im}\end{array}
\right]^T.
\end{equation} 
 
 In our controllability analysis, we utilize the sequences of these distance vectors, $D = \left(d^1, d^2, \hdots, d^{|D|}\right)$, where $|D|$ denotes the length of sequence $D$. In this representation, we drop the lower indices corresponding to the node labels, and use the super indices to denote the order of the particular vector in the sequence. In particular, we are interested in the sequences, $D$, defined by the following rule:

\textit{Rule}: For every $d^p \in D$, there exists an index, $k_p\in \{1,2,\cdots,m\}$ (where $m$ is the number of leaders), satisfying
\begin{equation}
\label{rule}
 [d^q]_{k_p}>[d^p]_{k_p}, \mbox{  $\forall q>p$.}  
\end{equation}

\textit{Example:} Consider a set of six vectors,
$$
\left\{\left[\begin{array}{cc}0\\3\\\end{array}\right] ,  \left[\begin{array}{cc}1\\2\\\end{array}\right], \left[\begin{array}{cc}1\\3\\\end{array}\right], \left[\begin{array}{cc}2\\1\\\end{array}\right], \left[\begin{array}{cc}2\\2\\\end{array}\right], \left[\begin{array}{cc}3\\0\\\end{array}\right]\right\}
$$

A vector sequence satisfying the rule in (\ref{rule}) can be

$$
D = \left(\left[\begin{array}{cc}\textcircled{0}\\3\\\end{array}\right] ,  \left[\begin{array}{cc}3\\\textcircled{0}\\ \end{array}\right], \left[\begin{array}{cc}2\\\textcircled{1}\\ \end{array}\right], \left[\begin{array}{cc}\textcircled{1}\\2\\ \end{array}\right], \left[\begin{array}{cc} \textcircled{2}\\2\\\end{array}\right]\right).
$$
For each vector $d^p$ in this sequence, the index $k_p$ satisfying the rule in (\ref{rule}) is marked with a circle. 
Note that here $d^1= \left[\begin{array}{cc}0\\3\\\end{array}\right]$  and $k_1 = 1$, as the first element of all other vectors $d^q$, where $q>1$, is greater than the first element of $d^1$ which is 0. Similarly, for the second vector in the sequence, $d^2 =\left[\begin{array}{cc}3\\0\\\end{array}\right]$, we have $k_2 = 2$, as the second element of all the vectors $d^q$ for $q>2$ are greater than the second element of $d^2$, and so on.

\begin{theorem}
 \label{lbound}

Let $\mathcal{D}$ be the set of all distance-vector sequences, D, satisfying the rule given in (\ref{rule}), and $|D^*|=\underset{D \in \mathcal{D}}{\textrm{max}}  \textrm{ $|D|$}$ be the maximum length for such sequences. Then the rank of the controllability matrix, $\Gamma$, satisfies
\begin{equation}
\label{lboundeq}
rank(\Gamma)\geq |D^*|.
\end{equation}
\end{theorem}
\begin{proof}
\newline
For a system with $n$ nodes, the controllability matrix is given as
\begin{equation}
\label{Gamma_M}
 \Gamma = \left[ \begin{array}{ccccc}
 B  & (-L)B  & (-L)^2B & \hdots & (-L)^{n-1}B \end{array}
\right] .
\end{equation}

Now, consider vectors of the form
\begin{equation}
\label{M_form}
 (-L)^{r_p}b_{k_p},
 \end{equation}
where $r_p=[d^p]_{k_p}$, and $b_{k_p}$ denotes the $k_p^{th}$ column of the input matrix $B$. Let $d^p$ be the distance vector of node $i$, i.e. $d^p=d_i$. Then, we have $r_p=[d_i]_{k_p}=d_{ik_p}$, and from Proposition \ref{prop1}, we know that the $i^{th}$ entry of the vector in (\ref{M_form}) is non-zero and equal to $[\mathcal{A}^{r_p}]_{ik_p}$. Also, for any node $j$ with $[d_{j}]_{k_i}>[d_{i}]_{k_i}$ we have the $j^{th}$ entry of this vector equal to zero. Using this along with the  sequence rule depicted in (\ref{rule}), we conclude that the $n \times |D^*|$ matrix 
\begin{equation}
\label{full_rank}
\left[ \begin{array}{ccccc}
  (-L)^{r_1}b_{k_1} &   (-L)^{r_2}b_{k_2}  &\hdots & (-L)^{r_{|D^*|}}b_{k_{|D^*|}} \end{array}
\right],   
\end{equation}
has full column rank since each column has a non-zero entry that none of the preceding columns have. Note that for every $p \in \{1,2,\hdots, |D^*|\}$, we have $r_p=[d^{p}]_{k_p} \leq n-1$ since the distance between any two nodes is always smaller than or equal to $n-1$. Hence, each column of the matrix in (\ref{full_rank}) is also a column of $\Gamma$, and rank of $\Gamma$ is greater than or equal to rank of the matrix in (\ref{full_rank}). Thus, we have $rank(\Gamma) \geq |D^*|$.

\end{proof}

\begin{figure*}[htb]
\begin{center}
\includegraphics[scale=1]{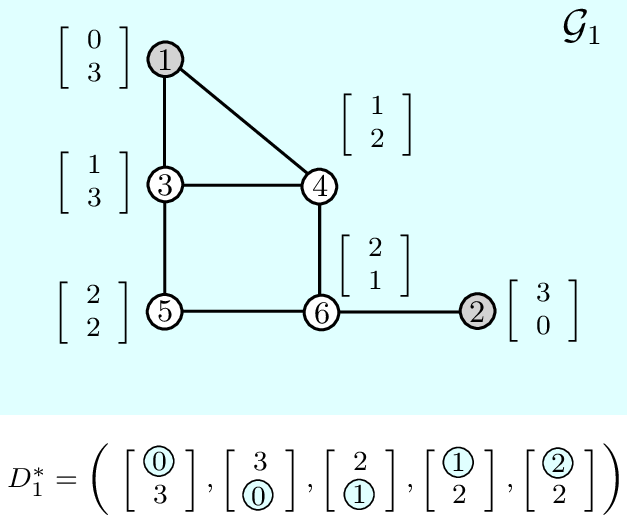}
\includegraphics[scale=1]{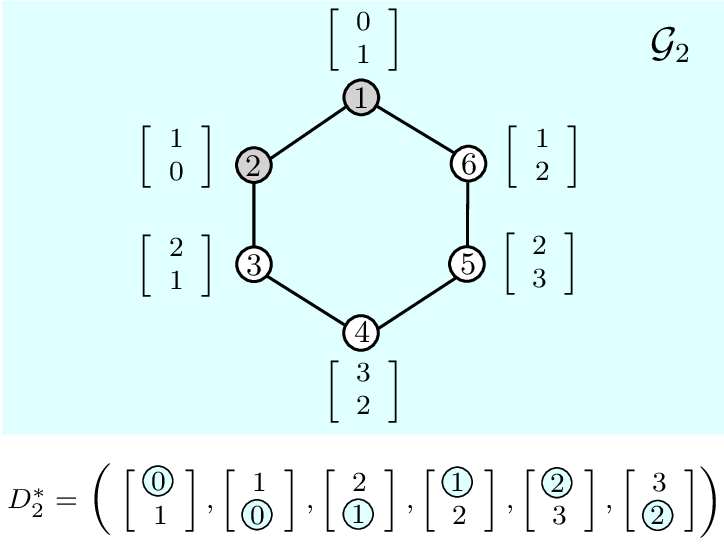}
\caption{Leader-follower networks, $\mathcal{G}_1$ and $\mathcal{G}_2$, each having two leaders namely, 1 and 2. Each node has its 2-dimensional distance vector (distances to the leaders) shown next to itself. For both networks, sample maximum length sequences, $D_1^*$ and $D_2^*$, satisfying the rule in (\ref{rule}) are given. For each vector in these sequences, the entry corresponding to the index satisfying the rule is circled.} 
\label{lbound_fig}
\end{center}
\end{figure*}

The lower bound presented in Theorem \ref{lbound} is tight and can not be improved for general graphs 
by only using the distances to the leaders. As a rather simple example, let us consider a network with a single leader. In that case, the distance vectors $d_i$ are one dimensional, hence the longest sequence satisfying the rule in (\ref{rule}) starts with $0$ and monotonically increases to the maximum distance from the leader. The length of this sequence is equal to the maximum distance plus one, which is equal to the number of cells in the distance partition with respect to the leader. Thus, for one dimensional case this lower bound is equal to the one presented in \cite{Zhang11}. A couple of examples with multiple leaders are depicted in Fig. \ref{lbound_fig}. For those networks, the lower bounds on the dimension of the controllable subspaces are computed as $|D_1^*|=5$, and $|D_2^*|=6$, whereas for both systems the actual ranks of the controllability matrices are equal to 6. Note that in general there is not a unique sequence with the maximum possible length, yet we present sample sequences, $D_1^*$ and $D_2^*$, in Fig \ref{lbound_fig}.

By combining the lower bound in Theorem \ref{lbound} and the upper bound in Theorem \ref{upperbound} we obtain the following corollary for the rank of the controllability matrix for any connected leader-follower network with the dynamics given in (\ref{sys}).  
 \begin{cor}
\label{bounds} 
Let $\mathcal{G}=(V,E)$ be a connected network with the dynamics given in (\ref{sys}). Let $|D^*|$ be the maximum sequence length for distance-vector sequences satisfying the rule in (\ref{rule}), and $\pi^*$ be the maximal leader-invariant EEP of $\mathcal{G}$. Then, the rank of the controllability matrix, $\Gamma$, satisfies
\begin{equation}
\label{boundseq}
 |D^*| \leq rank(\Gamma)\leq |\pi^*|.
 \end{equation}
\end{cor}

\section{Computing the lower bound}
\label{compute}

In this section we present an algorithm to compute the lower bound mentioned in the Theorem \ref{lbound}. Let $S=\{d_1, d_2,\hdots, d_n \}$ be the set of all distance-from-leaders vectors for a given graph. Given these distance vectors, let us consider a way of iteratively generating vector sequences satisfying the rule in (\ref{rule}). Let $C_p$ be the set of all distance vectors that can be assigned as the $p^{th}$ element of such a sequence $D$. According to these definitions, $C_1=S$. Once a vector from $C_p$ is assigned as the $p^{th}$ element of the sequence, $d^p$, and an index $k_p$ satisfying the sequence rule is chosen, $C_{p+1}$ can be obtained from $C_p$ as

\begin{equation}
\label{C_set}
 {C}_{p+1}= {C}_p \setminus \{d \in {C}_p \mid [d]_{k_p} \leq [d^p]_{k_p}\}.
  \end{equation}

In order to obtain longer sequences, this iteration must be continued until $C_p= \emptyset$.  However, in general there are too many possible sequences that can be obtained this way, and it is not feasible to find the maximum sequence length by searching among all these possibilities. Instead, we present a necessary condition for a sequence satisfying the rule in (\ref{rule}) to have the maximum possible length. This necessary condition significantly lowers the number of sequences that needs to be considered to find the maximum sequence length. 
\begin{prop}
 \label{algorithm}
Let $\mathcal{D}^*$ be a maximum length distance vector sequence satisfying the rule given in ( \ref{rule}), then its $p^{th}$ entry, $d^p$, satisfies 
\begin{equation}
\label{seq_eq}
[d^p]_{k_p}= \underset{d \in {C}_p}{\textrm{min}}  \textrm{ $[d]_{k_p}$}
\end{equation}
\end{prop}
\begin{proof}
Assume, for the sake of contradiction, this is not true. Then, there exists a distance vector $d_j \in  C_p$ such that $[d_j]_{k_p}<[d^p]_{k_p}$. Due to the rule of the sequence, $d_j$ can not be added to this sequence after $d^p$.  However, $d_j$ can be placed right before $d^p$ since its index $k_p$ satisfies the rule of the sequence.  Hence, we obtain a longer sequence satisfying the rule by placing $d_j$ right before $d^p$, which leads to the contradiction that $D^*$ does not have the maximum possible length.
\end{proof}

Note that in obtaining the lower bound, we only care about the lengths of sequences, not about their actual entries. Hence, if for any $d_i$, $d_j \in C_p $ we have $[d_i]_{k_p}= [d_j]_{k_p}= \underset{d \in C_p}{\textrm{min}}  \textrm{ $[d]_{k_p}$}$, then we do not care whether $d_i$ or $d_j$ is added to the sequence as $d^p$ since the resulting $C_{p+1}$ will be same as long as $k_p$ is chosen as the index satisfying the rule. Thus, as far as the sequence length is concerned, the only important decision at each step of the sequence generation is the choice of $k_p$. Based on this observation, we present an algorithm that can be used to compute the lower bound.

In this algorithm we define a new variable, $\mathcal{C}$, as the set of all possible non-empty sets $C_p$ that can be obtained at step $p$. Initially this set only includes the set of all the distance vectors, $S$, since there is a unique $C_1$ namely $S$. For each such $C_p$, one can obtain $m$ (number of leaders) different $C_{p+1}$ depending on the choice of $k_p$. Once, these $C_{p+1}$ are computed, we remove all the previous $C_p$ and  store the non-empty $C_{p+1}$ sets in $\mathcal{C}$, and continue the iteration. Iterations stop when $\mathcal{C}=\emptyset$. We keep a counter variable $\ell$ in the algorithm and it is incremented by one every time $\mathcal{C}$ is updated for the next step. Once we reach $\mathcal{C}=\emptyset$, the final value of $\ell$ gives us the maximum possible sequence length, $|D^*|$.

\begin{center}
\begin{tabular}{l l}
\rule[0.08cm]{8.5cm}{0.03cm}\\
\textbf{Algorithm I}\\
\rule[0.08cm]{8.5cm}{0.02cm}\\
\mbox{\small $\;1:\;$}\textbf{initialize:} $\mathcal{C}=\{S\}$ and $\ell=0$ \\
\mbox{\small $\;2:\;$}\textbf{while}\hspace{0.1cm}  $\mathcal{C}\ne \emptyset$ \\

\mbox{\small $\;3:\;$}\hspace{0.45cm} $\bar{\mathcal{C}}=\emptyset$\\

\mbox{\small $\;4:\;$}\hspace{0.45cm} \textbf{for} $i=1$ to $\mid\mathcal{C}\mid$\\

\mbox{\small $\;5:\;$}\hspace{1cm}   \textbf{for} $j=1$ to $m$\\

\mbox{\small $\;6:\;$}\hspace{1.35cm} $\bar{\mathcal{C}}_{(i-1)n_l + j} =  \mathcal{C}_i\setminus \{ d\in \mathcal{C}_i \mid [d]_j = \displaystyle{\min\limits_{d\in \mathcal{C}_i}}\;\; [d]_j$\}\\
\mbox{\small $\;7:\;$}\hspace{1cm} \textbf{end for}\\
\mbox{\small $\;8:\;$}\hspace{0.45cm} \textbf{end for}\\
\mbox{\small $\;9:\;$}\hspace{0.45cm} $\bar{\mathcal{C}} = \bar{\mathcal{C}} \setminus \{ C\in \bar{\mathcal{C}} \mid \;C=\emptyset$\}\\
\mbox{\small $10:\;$}\hspace{0.45cm} $\mathcal{C} = \bar{\mathcal{C}}$\\
\mbox{\small $11:\;$}\hspace{0.45cm} $\ell =\ell + 1$\\
\mbox{\small $12:\;$}\textbf{end while}\\
\mbox{\small $13:\;$}\textbf{return} \hspace{0.1cm} $\ell$\\
\rule[0.08cm]{8.5cm}{0.02cm}\\
\end{tabular}
\end{center}

For instance, consider the network $\mathcal{G}_1$ with two leaders shown in Fig. \ref{lbound_fig}. We can represent the flow of Algorithm I as a tree structure shown in the Fig. \ref{fig:Final12}. In this tree diagram, each node at a given level $p$ corresponds to an element of $\bar{\mathcal{C}}$ that is computed in the line 6 of Algorithm I in the $p^{th}$ iteration of the while loop. Algorithm will terminate after the fifth iteration of the while loop as all those $\bar{\mathcal{C}_i}$s will be empty sets.

\begin{figure*}[htb]
\begin{center}
\includegraphics[scale=0.9]{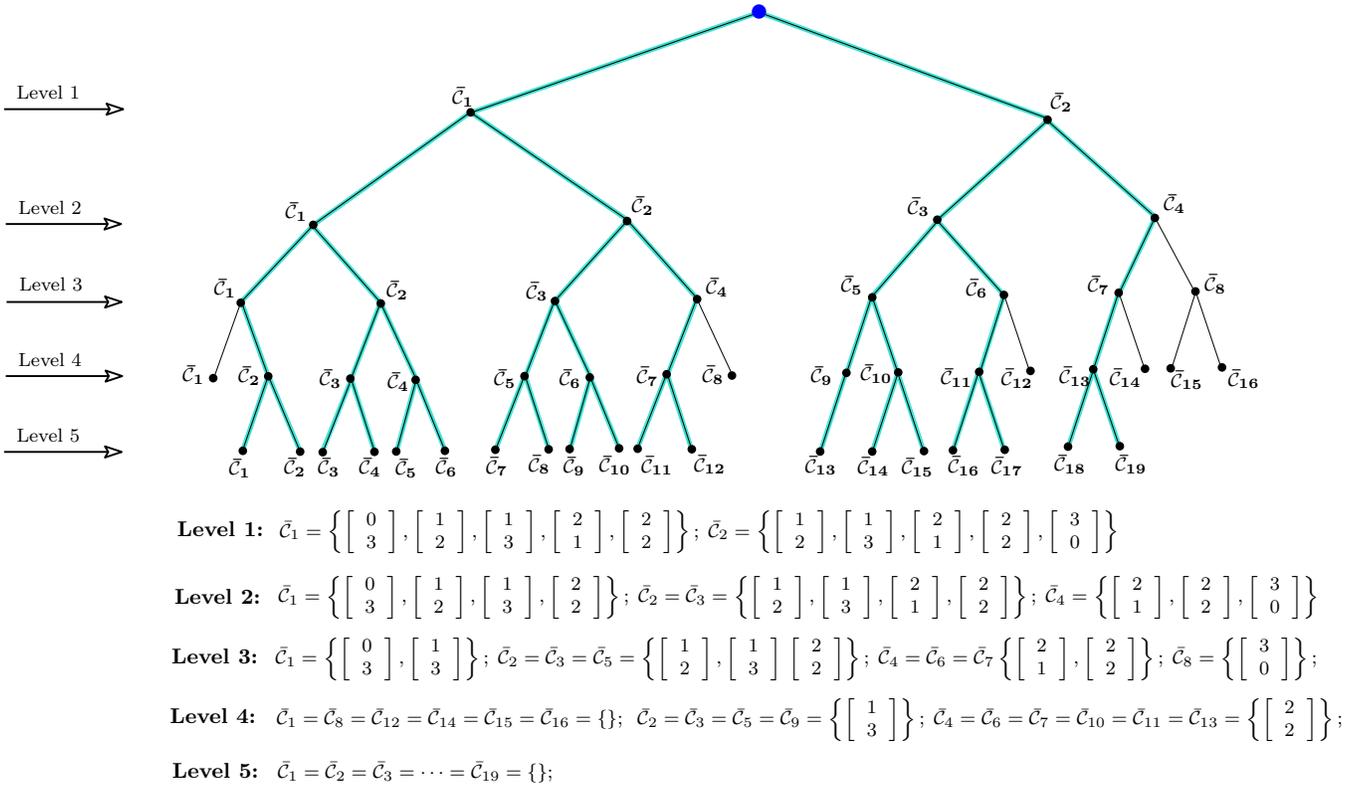}

\caption[Text excluding the matrix]{A tree representation for the flow of Algorithm I for the system $\mathcal{G}_1$ in Fig. \ref{lbound_fig}. Each time the \textit{while} loop is completed, the algorithm moves to the next level. Each node at the same level represents a particular $\bar{\mathcal{C}}_i$ computed in line 6 of the algorithm in the corresponding iteration of the while loop. The right child of a node $\bar{\mathcal{C}}_i$ corresponds to $\bar{\mathcal{C}}_i\setminus \{ d\in \bar{\mathcal{C}}_i \mid [d]_1 = \displaystyle{\min\limits_{d\in \bar{\mathcal{C}}_i}}\;\; [d]_1\}$, whereas, the left child corresponds to $\bar{\mathcal{C}}_i\setminus \{ d\in \bar{\mathcal{C}}_i \mid [d]_2 = \displaystyle{\min\limits_{d\in \bar{\mathcal{C}}_i}}\;\; [d]_2\}$. For example, $\bar{\mathcal{C}}_4$ at level 3 is a right child of $\bar{\mathcal{C}}_2$ at level 2, and is obtained by deleting all the vectors with the minimum first index from $\bar{\mathcal{C}}_2$ at level 2, namely $\left[\begin{array}{c} 1 \\ 2 \end{array}\right]$ and $\left[\begin{array}{c} 1 \\ 3 \end{array}\right]$. All $\bar{\mathcal{C}}_i$ are explicitly given below the tree diagram. In the fifth iteration of the while loop, each computed $\bar{\mathcal{C}}_i$ is an empty set and the algorithm terminates. The number of levels in the tree, stored in the variable $\ell$, corresponds to the required lower bound.}

\label{fig:Final12}
\end{center}
\end{figure*}

\section{Conclusion}
\label{conclusion}
In this paper we presented a graph theoretic analysis on the controllability of leader-follower networks with possibly multiple leaders. In particular, we presented a tight topological lower bound on the rank of the controllability matrix of such systems with arbitrary interaction graphs. This lower bound is based on the distances of nodes from the leaders. We also presented an algorithm to compute this lower bound for any leader-follower network. This lower bound may find its applications in various problems such as selecting leaders in a network that are sufficient to establish a certain level of controllability.

\end{document}